\documentclass{llncs}
\usepackage[utf8]{inputenc}
\usepackage[english]{babel}
\usepackage{amsmath}
\usepackage{amsfonts}
\usepackage{amssymb}
\usepackage{array}

\usepackage[linesnumbered,noend,noline]{algorithm2e}

\usepackage{tikz}
\newcommand{\tikznode}[2]{%
\ifmmode%
\tikz[remember picture,baseline=(#1.base),inner sep=0pt] \node (#1) {$#2$};%
\else
\tikz[remember picture,baseline=(#1.base),inner sep=0pt] \node (#1) {#2};%
\fi}

\usepackage{caption}
\newcommand{\argmin}{\operatorname{argmin}}
\newcommand{\rank}{\operatorname{rank}}
\newcommand{\sgn}{\operatorname{sgn}}

\newcommand{\glt}{\operatorname{GL}(K,n)^3}
\newcommand{\gl}{\operatorname{GL}(K,n)}
\renewcommand{\phi}{\varphi}

\begin{document}
\bibliographystyle{splncs04}

\title{A Normal Form for Matrix Multiplication Schemes}
\author{Manuel Kauers\thanks{M.K. was supported by the Austrian Science Fund (FWF) grant P31571-N32.}\orcidID{0000-0001-8641-6661} \and
  Jakob Moosbauer\thanks{J.M. was supported by the Land Oberösterreich through the LIT-AI Lab.}\orcidID{0000-0002-0634-4854}}
\institute{
  Institute for Algebra, Johannes Kepler University, Linz, Austria\\\relax
  \{manuel.kauers,jakob.moosbauer\}@jku.at
}

\maketitle

\begin{abstract}
	Schemes for exact multiplication of small matrices have a large symmetry group. This group defines an equivalence relation on the set of multiplication schemes. There are algorithms to decide whether two schemes are equivalent. However, for a large number of schemes a pairwise equivalence check becomes cumbersome. In this paper we propose an algorithm to compute a normal form of matrix multiplication schemes. This allows us to decide pairwise equivalence of a larger number of schemes efficiently.
\end{abstract}

\section{Introduction}

Computing the product of two $n \times n$ matrices using the straightforward algorithm costs $O(n^3)$ operations. Strassen found a multiplication scheme that allows to multiply two $2 \times 2$ matrices using only 7 multiplications instead of 8 \cite{St:Gein}. This scheme can be applied recursively to compute the product of $n\times n$ matrices in $O(n^{\log_2 7})$ operations. This discovery lead to a large amount of research on finding the smallest $\omega$ such that two $n\times n$ matrices can be multiplied using at most $O(n^\omega)$ operations. The currently best known bound is $\omega<2.37286$ and is due to Alman and Williams \cite{AW:ARLM}. 

Another interesting question is to find the exact number of multiplications needed to multiply two $n\times n$ matrices for small numbers $n$. For $n=2$ Strassen provided the upper bound of $7$. Winograd showed that we also need at least $7$ multiplications \cite{Wi:Omo2}. De Groote proved that Strassen's algorithm is unique \cite{dG:Ovoo} modulo a group of equivalence transformations.

For the case $n=3$ Laderman was the first to present a scheme that uses $23$ multiplications \cite{La:Anaf}, which remains the best known upper bound, unless the coefficient domain is commutative \cite{Ro:FCMA}. The currently best lower bound is $19$ and was proved by Bläser \cite{Bl:Otco}. There are many ways to multiply two $3\times 3$ matrices using $23$ multiplications \cite{JM:Nbaf,CBH:ANGM,OKM:Otio,Sm:Tbca,HKS:NWTM,BAD+:Epdo}.

For every newly found algorithm the question arises whether it is really new or it can be mapped to a known solution by one of the transformations described by de Groote. These transformations define an equivalence relation on the set of matrix multiplication algorithms. Some authors used invariants of the action of the transformation group to prove that their newly found schemes are inequivalent to the known algorithms. The works of Berger et al. \cite{BAD+:Epdo} and Heule et al. \cite{HKS:NWTM} provide algorithms to check if two given schemes are equivalent. Berger et al. give an algorithm that can check equivalence over the ground field $\mathbb{R}$ if the schemes fulfill a certain assumption. Heule et al. provide an algorithm to check equivalence over finite fields.

Heule et al. presented over 17,000 schemes for multiplying $3\times 3$ matrices and showed that they are pairwise nonequivalent, at least when viewed over the ground field $\mathbb{Z}_2$. Their collection has since been extended to over 64,000 pairwise inequivalent schemes. For testing whether a newly found scheme is really new, we would need to do an equivalence test for each of these schemes. Due to the large number of schemes this becomes expensive. 

In this paper we propose an algorithm that computes a normal form for the equivalence class of a given scheme over a finite field. If all known schemes already are in normal form, then deciding whether a newly found scheme is equivalent to any of them is reduced to a normal form computation for the new scheme and a cheap syntactic comparison to every old scheme. Although the transformation group over a finite field is finite, it is so large that checking equivalence by computing every transformation is not feasible. Thus, Heule et al. use a strategy that iteratively maps one scheme to another part by part. We use a similar strategy to find a minimal element of an equivalence class. 

\section{Matrix Multiplication Schemes}
Let $K$ ba field and let $\mathbf{A},\mathbf{B} \in K^{n\times n}$. The computation of the matrix product $\mathbf{C}=\mathbf{AB}$ by a Strassen-like algorithm proceeds in two stages. In the first stage we compute some intermediate products $M_1, \ldots, M_r$ of linear combinations of entries of $\mathbf{A}$ and linear combinations of entries of $\mathbf{B}$. In the second stage we compute the entries of $\mathbf{C}$ as linear combinations of the $M_i$. 

For example if $n=2$, we can write

\begin{equation*}
\mathbf{A} = \begin{pmatrix}
   a_{1,1}&a_{1,2}\\
   a_{2,1}&a_{2,2}\\
 \end{pmatrix}
 \quad
 \mathbf{B} = \begin{pmatrix}
   b_{1,1}&b_{1,2}\\
   b_{2,1}&b_{2,2}
 \end{pmatrix}
 \quad\text{and}\quad
 \mathbf{C} = \begin{pmatrix}
   c_{1,1}&c_{1,2}\\
   c_{2,1}&c_{2,2}
 \end{pmatrix}.
\end{equation*}

Strassen's algorithm computes $\mathbf{C}$ in the following way:
\begin{eqnarray*}
	M_1 &=& (a_{1,1} + a_{2,2}) (b_{1,1} + b_{2,2})\\
	M_2 &=& (a_{2,1} + a_{2,2}) (b_{1,1})\\
	M_3 &=& (a_{1,1}) (b_{1,2} - b_{2,2})\\
	M_4 &=& (a_{2,2}) (b_{2,1} - b_{1,1})\\
	M_5 &=& (a_{1,1} + a_{1,2})(b_{2,2})\\
	M_6 &=& (a_{2,1} - a_{1,1}) (b_{1,1}+ b_{1,2})\\
	M_7 &=& (a_{1,2} - a_{2,2}) (b_{2,1} + b_{2,2})
\end{eqnarray*}
\begin{eqnarray*}
	c_{1,1} &=& M_1 + M_4 - M_5 + M_7\\
	c_{1,2} &=& M_3 + M_5\\
	c_{2,1} &=& M_2 + M_4\\
	c_{2,2} &=& M_1 - M_2 + M_3 + M_6.
\end{eqnarray*}

A Strassen-like multiplication algorithm that computes the product of two $n\times n$ matrices using $r$ multiplications has the form

\begin{eqnarray*}
M_1 &=& (\alpha_{1,1}^{(1)}a_{1,1} + \alpha_{1,2}^{(1)}a_{1,2}+\cdots)(\beta_{1,1}^{(1)}b_{1,1} + \beta_{1,2}^{(1)}b_{1,2}+\cdots)\\
&\vdots&\\
M_r &=& (\alpha_{1,1}^{(r)}a_{1,1} + \alpha_{1,2}^{(r)}a_{1,2}+\cdots)(\beta_{1,1}^{(r)}b_{1,1} + \beta_{1,2}^{(r)}b_{1,2}+\cdots)\\
c_{1,1} &=& \gamma_{1,1}^{(1)}M_1 + \gamma_{1,1}^{(2)}M_2 + \cdots + \gamma_{1,1}^{(r)}M_{r}\\
&\vdots&\\
c_{n,n} &=& \gamma_{n,n}^{(1)}M_1 + \gamma_{n,n}^{(2)}M_2 + \cdots + \gamma_{n,n}^{(r)}M_{r}.
\end{eqnarray*}

All the information about such a multiplication scheme is contained in the coefficients $\alpha_{i,j},\beta_{i,j}$ and $\gamma_{i,j}$. We can write these coefficients as a tensor in $K^{n\times n} \otimes K^{n\times n} \otimes K^{n\times n}$: 

\begin{equation}\label{tensor}
	\sum_{l=1}^r ((\alpha_{i,j}))_{i=1,j=1}^{n,n} \otimes ((\beta_{i,j}))_{i=1,j=1}^{n,n} \otimes ((\gamma_{i,j}))_{i=1,j=1}^{n,n}.
\end{equation}

As an element of $K^{n\times n} \otimes K^{n\times n} \otimes K^{n\times n}$ a correct scheme is equal to the matrix multiplication tensor defined by $\sum_{i,j,k = 1}^n E_{i,k}\otimes E_{k,j} \otimes E_{i,j}$ where $E_{u,v}$ is the matrix with $1$ at position $(u,v)$ and zeros everywhere else. Formulas become a bit more symmetric if we look at the tensor $\sum_{i,j,k = 1}^n E_{i,k}\otimes E_{k,j} \otimes E_{j,i}$ corresponding to the product $\mathbf{C}^T = \mathbf{AB}$, so we will consider this tensor instead.

We represent a scheme as a table containing the matrices in this tensor. We will refer to the rows and columns of this table as the rows and columns of a scheme. For example Strassen's algorithm is represented as shown in Table~\ref{strassen}.

\begin{table}
\begin{center}
\footnotesize
\begin{tabular}{c|c|c}
 $\left(
\begin{array}{cc}
 1 & 0 \\
 0 & 1 \\
\end{array}
\right)$ & $\left(
\begin{array}{cc}
 1 & 0 \\
 0 & 1 \\
\end{array}
\right)$ & $\left(
\begin{array}{cc}
 1 & 0 \\
 0 & 1 \\
\end{array}
\right)$ \rule[-3ex]{0pt}{7ex}\\
\hline
 $\left(
\begin{array}{cc}
 0 & 0 \\
 1 & 1 \\
\end{array}
\right)$ & $\left(
\begin{array}{cc}
 1 & 0 \\
 0 & 0 \\
\end{array}
\right)$ & $\left(
\begin{array}{cc}
 0 & 0 \\
 1 & -1 \\
\end{array}
\right)$  \rule[-3ex]{0pt}{7ex}\\
\hline
$\left(
\begin{array}{cc}
 1 & 0 \\
 0 & 0 \\
\end{array}
\right)$ & $\left(
\begin{array}{cc}
 0 & 1 \\
 0 & -1 \\
\end{array}
\right)$ & $\left(
\begin{array}{cc}
 0 & 1 \\
 0 & 1 \\
\end{array}
\right)$\rule[-3ex]{0pt}{7ex}\\
\hline
 $\left(
\begin{array}{cc}
 0 & 0 \\
 0 & 1 \\
\end{array}
\right)$ & $\left(
\begin{array}{cc}
 1 & 0 \\
 -1 & 0 \\
\end{array}
\right)$ & $\left(
\begin{array}{cc}
 1 & 0 \\
 1 & 0 \\
\end{array}
\right)$ \rule[-3ex]{0pt}{7ex}\\
\hline
$\left(
\begin{array}{cc}
 1 & 1 \\
 0 & 0 \\
\end{array}
\right)$ & $\left(
\begin{array}{cc}
 0 & 0 \\
 0 & 1 \\
\end{array}
\right)$ & $\left(
\begin{array}{cc}
 -1 & 1 \\
 0 & 0 \\
\end{array}
\right)$ \rule[-3ex]{0pt}{7ex}\\
\hline
$\left(
\begin{array}{cc}
 -1 & 0 \\
 1 & 0 \\
\end{array}
\right)$ & $\left(
\begin{array}{cc}
 1 & 1 \\
 0 & 0 \\
\end{array}
\right)$ & $\left(
\begin{array}{cc}
 0 & 0 \\
 0 & 1 \\
\end{array}
\right)$ \rule[-3ex]{0pt}{7ex}\\
\hline
$\left(
\begin{array}{cc}
 0 & 1 \\
 0 & -1 \\
\end{array}
\right)$ & $\left(
\begin{array}{cc}
 0 & 0 \\
 1 & 1 \\
\end{array}
\right)$ & $\left(
\begin{array}{cc}
 1 & 0 \\
 0 & 0 \\
\end{array}
\right)$\rule[-3ex]{0pt}{7ex}
\end{tabular}
\end{center}
\hrule
\caption{Strassen's Algorithm}
\label{strassen}
\end{table}

\section{The Symmetry Group}
There are several transformations that map one correct matrix multiplication scheme to another one. We call two schemes equivalent if they can be mapped to each other by one of these transformations. De Groote \cite{dG:Ovoo} first described the transformations and showed that Strassen's algorithm is unique modulo this equivalence. 

The first transformation is permuting the rows of a scheme. This corresponds to just changing the order of the $M_i$'s in the algorithm. 
Another transformation comes from the fact that $\mathbf{AB}=\mathbf{C}^T \Leftrightarrow \mathbf{B}^T \mathbf{A}^T = \mathbf{C}$. It acts on a tensor by transforming a summand $A\otimes B \otimes C$ to $B^T\otimes A^T \otimes C^T$. Moreover, it follows from the condition that the sum (\ref{tensor}) is equal to the matrix multiplication tensor, that also a cyclic permutation of the coefficients $\alpha, \beta$ and $\gamma$ is a symmetry transformation. Taking those together we get an action that is composed by an arbitrary permutation of the columns of a scheme and transposing all the matrices if the permutation is odd.

Finally, we can use that for any invertible matrix $V$ we have $\mathbf{AB}=\mathbf{A}VV^{-1}\mathbf{B}$. The corresponding action on a tensor $A \otimes B \otimes C$ maps it to $AV \otimes V^{-1}B \otimes C$. Since we can permute $A$, $B$ and $C$ we also can insert invertible matrices $U$ and $W$ which results in the action
\begin{equation*}
(U,V,W)*A\otimes B\otimes C = UAV^{-1}\otimes VBW^{-1} \otimes WCU^{-1}.
\end{equation*}
This transformation is called the sandwiching action.

If we combine all these transformations we get the group $G=S_r \times S_3 \ltimes \glt$ of symmetries of $n\times n$ matrix multiplication schemes with $r$ rows. 

\begin{definition}
Let $\phi \colon S_3 \to \operatorname{Aut}(\glt)$ be defined by 
\begin{equation*}
\phi(\pi) = \begin{cases} 
(U,V,W)\mapsto\pi((U,V,W)) &\quad\text{if }\sgn(\pi)=1\\ 
(U,V,W)\mapsto\pi((V^{-T},W^{-T},U^{-T})) &\quad\text{if }\sgn(\pi)=-1
\end{cases}
\end{equation*}
The symmetry group of $n\times n$ matrix multiplication schemes with $r$ rows is defined over the set $G=S_r\times S_3 \times \glt$ with the multiplication given by 
\begin{eqnarray*}
(\sigma_1,\pi_1,(U_1,V_1,W_1))\cdot (\sigma_2,\pi_2,(U_2,V_2,W_2)) = \\
(\sigma_1 \sigma_2,\pi_1 \pi_2,(U_1,V_1,W_1)\phi(\pi_1)((U_2,V_2,W_2))).
\end{eqnarray*}

The action $g*S$ of a group element $g=(\sigma,\pi,U,V,W) \in G$ on a multiplication scheme $S\in (K^{n\times n})^{r\times 3}$ is defined by first letting $\sigma$ permute the rows of $S$ then letting $\pi$ permute the columns of $S$ and transposing every matrix if $\sgn(\pi)=-1$ and finally letting $U,V$ and $W$ act on every row as described above. 
\end{definition}

One can show that this action fulfills the criteria of a group action.

\section{Minimal Orbit Elements}

Two schemes are equivalent if they belong to the same orbit under the action of the group $G$. Our goal in this section is to define a normal form for every orbit. The particular choice of the normal form is partly motivated by implementation convenience and not by any special properties. 
From now on we assume that $K$ is a finite field. Since over a finite field the symmetry group is finite we could decide equivalence or compute a normal form by exhaustive search. However, already for $n=3$ the symmetry group over $\mathbb{Z}_2$ has a size of $23!\cdot 6\cdot 4741632$.

\begin{definition}
	Let $S \in (K^{n\times n})^{r\times 3}$ be a matrix multiplication scheme. The rank pattern of the scheme is defined to be the table \[((\rank S_{i,1},\rank S_{i,2},\rank S_{i,3}))_{i=1}^r\] and the rank vector of a row $(A,B,C)$ to be $(\rank(A),\rank(B),\rank(C))$.
\end{definition}

Since the matrices $U, V$ and $W$ are invertible, the sandwiching action leaves the rank pattern invariant. Transposing the matrices does not change their rank either. Therefore the only way a group element changes the rank pattern of a scheme is by permuting it accordingly. So for two equivalent schemes their rank patterns only differ by a permutation of rows and columns. This allows us to permute the rows and columns of the scheme such that the rank pattern becomes maximal under lexicographic order. 

This maximal rank pattern is a well-known invariant of the symmetry group that has been used to show that two schemes are not equivalent. For example Courtois et al. \cite{CBH:ANGM} and Oh et al. \cite{OKM:Otio} used this test to prove that their schemes were indeed new. However, this method only provides a sufficient condition for the inequivalence of schemes and can not decide equivalence of schemes. In Heule et al.'s data for certain rank patterns there almost 1000 inequivalent schemes having this rank pattern.

We choose the normal form to be an orbit element which has a maximal rank pattern and is minimal under a certain lexicographic order. For doing so fix a total order on $K$ such that $0<1<x$ for $x\in K\setminus\{0,1\}$. The order need not be compatible with $+$ or $\cdot$ in any sense. For the matrices in the schemes we use colexicoraphic order by columns, with columns compared by lexicographic order. This means for two column vectors $v=(x_1, \ldots, x_n)^T$ and $v'=(x_1',\ldots,x_n')^T$ we define recursively
\begin{equation*}
v < v' :\Leftrightarrow x_1<x_1' \vee (x_1 = x_1' \wedge (x_2,\ldots x_{n}) < (x_2',\ldots,x_{n}'))
\end{equation*}
For two matrices $M=(v_1\mid\cdots \mid v_n)$ and $M'=(v_1'\mid\cdots \mid v_n')$ we define
\begin{equation*}
M < M' :\Leftrightarrow v_n<v_n' \vee (v_n = v_n' \wedge (v_1\mid\cdots\mid v_{n-1}) < (v_1'\mid \cdots \mid v_{n-1}'))
\end{equation*}


For ordering the schemes we use the common lexicographic order. So we compare two schemes row by row from top to bottom and in each row we compare the matrices from left to right using the order defined above. 

\begin{definition}
Let $S\in (K^{3\times 3})^{r\times 3}$ be a matrix multiplication scheme. We say $S$ is in normal form if $S = \min\{S' \in G*S \mid \text{the rank pattern of } S' \text{is sorted}\}$, where the minimum is taken with respect to the order defined above.
\end{definition}

Such a normal form clearly exists and it is unique since the group $G$ is finite and the lexicographic order is a total order.

The strategy to compute the normal form is as follows:
We first consider the rank pattern of a scheme and apply row and column permutations that maximize the rank pattern. If there are several column permutations that lead to the same rank pattern, we consider each of them separately, since there are at most six. 

Then we proceed row by row. For all rows that have maximal rank pattern, we determine the minimal element of its orbit under the action of $\glt$. From the definition of the normal form, it follows that the smallest row we can produce this way has to be the first row in the normal form. However, we might be able to reach that row from several different rows and also the choice of $U,V$ and $W$ is in general not unique. 

Apart from the first row of the normal form we also compute the stabilizer of the first row, which is the set of all triples $(U,V,W)\in \glt$ such that $(A,B,C) = (UAV^{-1},VBW^{-1},WCU^{-1})$. For each possible row that can be mapped to the first row we compute the \emph{tail}, by which we mean list of all remaining rows after applying suitable triple $(U,V,W)$. 

We then continue this process iteratively. We go over each tail and determine a row that has maximal rank vector and becomes minimal under the action of the stabilizer. To do this we apply every element of the stabilizer to all possible candidates for the next row. This uniquely determines the next row of the normal form and we get again a list of tails and the stabilizer of the already determined rows.

The full process is listed in Algorithm \ref{alg1}.

\begin{figure}[h]
\begin{algorithm}[H]
		\SetKwInOut{Input}{Input}\SetKwInOut{Output}{Output}
		\Input{A matrix multiplication scheme $s$}
		\Output{An equivalent scheme in normal form}
		$P:= \{\pi*s \mid \pi \in S_3 \wedge \textit{ there is a row permutation of }\pi*s\textit{ with maximal rank pattern}\}$\\
		$n:=s$\\
		\For{$s' \in P$}{
			$\mathit{candidate}:=()$\\
			$\mathit{tails} = \{s'\}$\\
			$\mathit{stab} = \glt$\\
			\While{$tails\neq \{()\}$}{
				$min := 1$\\
				\For{$t\in \mathit{tails}$}{
					\For{$r \in t$ with maximal rank vector}{	
						$g := \argmin_{g\in \mathit{stab}}g*r$\\
						\If{$g*r < min$}{
							$min:= g*r$\\
							$\mathit{newtails}:=\{\}$;
						}
						\If{$g*r = min$}{
							$\mathit{newtails} := \mathit{newtails} \cup (g*t\setminus \{g*r\})$
						}
					}	
				}
				$\mathit{stab} := \{g\in \mathit{stab}\mid g*min=min\}$\\
				$\mathit{tails} := \mathit{newtails}$\\
				append $min$ to $\mathit{candidate}$
			}
			\If{$\mathit{candidate}<n$}{
				$n:=\mathit{candidate}$			
			}
		}
		\Return $n$		
\end{algorithm}
\hrule
\caption{Normal Form Computation}\label{alg1}
\end{figure}

\begin{proposition}
Algorithm 1 terminates and is correct.
\end{proposition}
\begin{proof}
The termination of the algorithm is guaranteed, since in line~12 the new candidates contain one row less than in the previous step, so eventually the list of tails only contains empty elements.

To prove correctness we first note that the choice of $P$ ensures that it contains a scheme that can be mapped to its normal form without applying further column permutations. From now on we only consider the iteration of the loop in line~3 where $s'$ is this scheme. 

It remains to show that after lines~4 to 19 the candidate is in normal form. To this end we prove the following loop invariant for the while loop: \textit{candidate} is equal to the first rows of the normal form and there is a $g\in\mathit{stab}$ and a $t \in \mathit{tails}$ such that $g*t$ is a permutation of the remaining rows of the normal form.

The lines~4,~5 and 6 ensure that the loop invariant is true at the start of the loop. We now assume that the loop invariant holds at the beginning of an iteration and prove that it is still true after the iteration. Since we know that there are $g\in\mathit{stab}$ and $t \in \mathit{tails}$ such that $g*t$ contains the next row of the normal form and the rank vector is invariant under the group action the lines~9 and 10 will at some point select an $r$ that can be mapped to the next row of the normal form. 

Since the normal form is the lexicographically smallest scheme in its equivalence class, the next row must always be the smallest row that has not been added to $\mathit{candidate}$ yet. Therefore by choosing $g$ such that $g*r$ is minimal in line~11 we ensure that $\mathit{min}$ is the next row of the normal form. For every element $t$ of $\mathit{tails}$ that contains a row that can be mapped to the next row of the normal form we add a suitably transformed copy of $t$ with this row removed to $\mathit{newtails}$. Since $\mathit{tails}$ contains at least one element that can be mapped to the remaining rows of the normal form and we discard only elements that cannot have this property, it is guaranteed that $\mathit{newtails}$ still has this property after lines~9 to 16. 

Finally, we have to show that $\mathit{stab}$ still contains a suitable element. Let $t\in \mathit{tails}$ and $g'\in\mathit{stab}$ be such that $g'$ maps $t$ to a permutation of the remaining rows of the normal form. Let $g$ be the element chosen to minimize $r$ in line~11. Since $\mathit{stab}$ is a group it must contain $g'\cdot g^{-1}$.  Moreover, $\mathit{newtails}$ contains $g*t\setminus\{g*r\}$ which is mapped to $g'*t\setminus \{g'*r\}$. Therefore, $g'\cdot g^{-1}$ has the desired property. \qed
\end{proof}

\section{Minimizing the First Row}

Algorithm~\ref{alg1} is more efficient than a naive walk through the whole symmetry group $G$ because we can expect the stabilizer to become small during the computation. However, in the first iteration we still go over the full group $\glt$. In this section we describe how this can be avoided.

The order we have chosen ensures that the first row has a particular form.

\begin{proposition} \label{form}
Let $G=\glt$ and let $(A,B,C)\in (K^{n\times n})^3$ be such that $(A,B,C)$ is the minimal element of $G*(A,B,C)$. Then the following hold:
\begin{enumerate}
\item $A$ has the form 
\begin{equation}
\label{subdiagonal}
\begin{pmatrix}
0 & 0\\
I_r & 0
\end{pmatrix}
\end{equation}

where $r=\rank A$.
%
In particular, if $\rank A = n$, then $A=I_n$.

\item $B$ is in column echelon form.  
\end{enumerate}
\end{proposition}
\begin{proof}
Using Gaussian elimination we can find $(A',B',C')=(U,V,W)*(A,B,C)$ where $A'$ and $B'$ are in the described form. To show that $(A,B,C)$ already are in this form we proceed by induction on $n$. If $n=1$, then the claims are true. For the induction step assume that the claims are true for $n-1$. 
\begin{enumerate}
\item We first consider the special case $\rank A = n$. Denote by $v_1,\ldots ,v_n$ the columns of $A$. Since $A\leq A' = I_n$ there are two cases:

Case 1: $v_n < e_n$. Then $v_n = 0$ contradicting the assumption that $A$ has full rank.

Case 2: $v_n = e_n$
. Then the last row of $A$ contains only zeros apart from the $1$ in the bottom right corner. Otherwise we could use column reduction to make $A$ smaller. Since $A$ is minimal, also the matrix we get when we remove the last column and row from $A$ has to be minimal. So by the induction hypothesis $A$ has the desired form. 

Now suppose $\rank A < n$. Since the last column of $A'$ contains only zeros and $A$ is minimal, the last column of $A$ consists only of zeros. We can use row reduction to form a matrix $A''$ that is equivalent to $A$, has a zero row and all other rows equal to those of $A$. So $A''<A$. We then shift the zero row of $A''$ to the top. Since this doesn't make $A''$ bigger, it is still not greater than $A$. Because of the minimality of $A$ its first row has then to be zero as well. Now we can remove the last column and first row of $A$ and the resulting matrix must still be minimal. So by the induction hypothesis $A$ is of the desired form. 

\item Since we already showed $A = A'$ we can assume $U=V=I_n$. So $B'$ is the column echelon form of $B$. 
We write $B$ as $(v_1\mid \cdots \mid v_n)$ and $B'$ as $(v_1'\mid \cdots \mid v_n')$. We again have two cases:

Case 1: $v_n < v_n'$. So $v_n'\neq 0$ and since $B'$ is in column echelon form this implies $v_n' = e_n$. Then $v_n=0$ which contradicts that $B'$ is the column echelon form of $B$.

Case 2: $v_n = v_n'$. Since $B'$ is in column echelon form we either have $v_n=e_n$ or $v_n = 0$. We claim that the matrix we get by removing the last column and row from $B$ is minimal. If not, there is a sequence of column operations that makes that matrix smaller. Let $B''=(v_1''\mid\cdots\mid v_n'')$ be the matrix we get by applying these operations to $B$ and let $i$ be the index of the right most column that was changed. So $v_i''$ with the last element removed must be smaller than $v_i$ with the last element removed. However, this implies that $v_i''<v_i$ and therefore $B''<B$, which is a contradiction. So by the induction hypothesis $B$ with the last row and column removed must be in column echelon form. 

It remains to show that the last rows of $B$ and $B'$ are equal. There must exist a sequence of column operations that turn $B$ into $B'$. If $v_n=e_n=v_n'$, then these operations would eliminate all elements in the last row of $B$, except the one in the bottom right corner. This implies $B'\leq B$ and therefore $B'=B$. If $v_n = 0 = v_n'$, then this sequence cannot change the last row because any column operation not involving the last column would destroy the column echelon form in the upper left part. Therefore $B=B'$. \qed
\end{enumerate}
\end{proof}

\begin{figure}
\begin{algorithm}[H]
		\SetKwInOut{Input}{Input}\SetKwInOut{Output}{Output}
		\Input{A triple of $n \times n$ matrices $(A,B,C)$}
		\Output{A minimal triple equivalent under the action of $\glt$}
		
		\If{$\rank A = n$}{
			$A_1:=I_n$\\
			$C':=CA$\\
			\If{$\rank B = n$}{
				$B_1:=I_n$\\
				$C_1:=\min_{W\in \gl} W B^{-1}C'W^{-1}$
			}
			\Else{
				$B_1:=\min_{V,W\in \gl} VBW^{-1}$\\
				$S := \{(V,W) \mid V,W\in \gl \wedge VBW^{-1}=B_1 \}$\\
				$C_1:=\min_{(V,W)\in S} WC'V^{-1}$
			}
		}
		\Else{
			$A_1:=\min_{U,V\in \gl} UAV^{-1}$\\
			$(U,V):=\argmin_{U,V\in \gl} UAV^{-1}$\\
			$B' = VB; C' = CU^{-1}$\\
			$S := \{(U,V) \mid U,V\in \gl \wedge UA_1V^{-1}=A_1 \}$\\
			$B_1:=\min_{(U,V)\in S,W\in \gl} VB'W^{-1}$\\
			$C'' = WC'$, where $W$ is chosen as in the line above\\
			$S' := \{(U,V,W)\in\glt \mid UA_1V^{-1}=A_1 \wedge VB_1W^{-1}=B_1\}$\\
			$C_1:=\min_{(U,V,W)\in S'} WC''V^{-1}$
		}
\end{algorithm}
\hrule
\caption{Special treatment of first row}
\label{alg2}
\end{figure}

Let $S\in (K^{n\times n})^{r\times 3}$ be a matrix multiplication scheme. Denote by $(U,V,W)$ the element of $\glt$ used to transform $S$ into normal form and denote by $(A_1,B_1,C_1)$ the first row of the normal form of $S$. Let $(A,B,C)$ be the row that is mapped to $(A_1,B_1,C_1)$ and assume that the columns of $S$ do not need to be permuted. 

Then $(A,B,C)$ must have a maximal rank vector. Therefore, $A$ has the maximal rank of all the matrices in the scheme. So if the scheme contains a matrix of full rank then $A$ has full rank. Moreover, $A_1$ is the minimal element equivalent to $A$ under the action of $\glt$. 

If $A$ has full rank, then $A_1=I_n$ by Proposition~\ref{form}. So we consider the scheme $S'=(A^{-1},I_n,I_n)*S$ instead and update $A,B,C$ and $U,V,W$ accordingly. Then $A=A_1=I_n$ and therefore $U=V$. Since there is no restriction on the choice of $V$ it follows from the first part of Proposition~\ref{form} that $B_1$ must be of the form~(\ref{subdiagonal}). 

If $B$ also has full rank, then we set $S''=(B^{-1},B^{-1},I_n)*S'$ and adjust $A,B,C$ and $U,V,W$ again. So we have $B=B_1=I_n$ and $U=V=W$. Now we can determine $C_1$ and the stabilizer of the first row by iterating over $\gl$ and minimizing $WCW^{-1}$.

If $B$ does not have full rank, we determine all invertible matrices $V$ and $W$ such that $VBW^{-1}=B_1$. This can be done by solving the linear system $VB = WB_1$ and discarding all singular solutions. Since $U=V$ we go through all possibilities for $V$ and $W$ and minimize $WCV^{-1}$. This allows us to determine $C_1$ and the stabilizer of the first row.

If $A$ does not have full rank, we solve the linear system $UA=VA_1$ and discard all singular solutions. The remaining solutions are the possible choices for $U$ and $V$ such that $UAV^{-1}=A_1$. By Proposition~\ref{form} $B_1$ must be in column echelon form. So for all possible choices of $U$ and $V$ we determine $W$ such that $VBW^{-1}$ is in column echelon form. The smallest matrix $VBW^{-1}$ constructed this way must be equal to $B_1$. Then we go over all such triples $(U,V,W)$ that map $B$ to $B_1$ and determine those that minimize $WCU^{-1}$. So we find $C_1$ and the stabilizer of the first row.

The process is summarized in Algorithm~\ref{alg2}.

\section{Timings}

For a comparison we have tested the equivalence check of Heule et al. on 10,000 randomly selected pairs from Heule et al.'s data set of $3\times 3$ multiplication schemes and computed the normal form of 10,000 randomly selected schemes. Checking equivalence of two schemes took on average $0.0092$ seconds. Computing a normal form took on average $1.87$ seconds. Thus, in our application checking equivalence of a single new scheme against a set of known schemes in normal form is faster than directly checking equivalence as soon as we have at least 204 schemes. 

\bibliography{normalforms}

\end{document}